\documentclass[pra, twocolumn]{revtex4-1}
\usepackage{graphicx}
\usepackage{hyperref}
\usepackage[section]{placeins}
\usepackage{amsmath}
\usepackage{amssymb}
\usepackage{amsthm}
\usepackage{dsfont}
\usepackage{bbold}
\usepackage[caption=false]{subfig}
\usepackage{verbatim}
\newtheorem*{lemma*}{Lemma (Ideal detection)}
\newtheorem*{theorem*}{Theorem (Noisy detection)}
\newtheorem*{corollary*}{Corollary (Noiseless detection)}

\begin{document}

\title{Theory of multiplexed photon number discrimination}

\author{Filippo.~M.~Miatto$^{1}$, Akbar Safari$^{1}$, Robert W.~Boyd$^{1,2}$}
\affiliation{$^1$Dept.~of Physics, University of Ottawa, 150 Louis Pasteur, Ottawa, Ontario, K1N 6N5 Canada}
\affiliation{$^2$Institute of Optics, University of Rochester, Rochester, USA}
\date{\today}

\begin{abstract}
Although some non-trivial photon number resolving detectors exist, it may still be convenient to discriminate photon number states with the method of multiplexed detection. Multiplexing can be performed with paths in real space, with paths in time, and in principle with any degree of freedom that has a sufficient number of eigenstates and that can be coupled to the photon number. Previous works have addressed the probabilities involved in these measurements with Monte Carlo simulations, or by restricting the number of detectors to powers of 2, or without including quantum efficiency or noise.
In this work we find an analytical expression of the detection probabilities for any number of input photons and any number of on/off photon detectors with a quantum efficiency $0\%\leq\eta\leq100\%$ and a false count probability $\varepsilon\geq0$. This allows us to retrodict the number of photons that we had at the input in the least unbiased way possible. We conclude our work with some examples.
\end{abstract}
\maketitle

\section{Introduction}
For practical applications of quantum optics it would be a great advantage to have a detector that can discriminate between different photon number states \cite{silberhorn2007detecting}. There are currently several different solutions that allow one to achieve this to some extent \cite{miller2003demonstration, fujiwara2005multiphoton, gansen2007photon, kardynal2007photon, divochiy2008superconducting, namekata2010non}, but the resources that such detectors require (such as very low temperatures, costly materials and/or optical configurations) make them rather costly to obtain and operate. There are workarounds that involve squeezing more information out of the conventional detectors \cite{kim1999multiphoton,kardynal2008avalanche}, or by multiplexing the photons towards multiple single-photon detectors \cite{achilles2003fiber, fitch2003photon, achilles2004photon, jiang2007photon, lantz2008multi,avenhaus2010accessing,thomas2012practical}.

The most common single-photon detectors are only able to tell us whether they detected ``zero photons'' or ``more than zero photons''. Furthermore they are subject to noise and a sub-optimal efficiency, which means that sometimes they click when they shouldn't have clicked or that they fail to click when they should have clicked \cite{hadfield2009single}.

In this work we explore photon-number discrimination by multiplexing, our novel contribution is to take into account quantum efficiency and noise, as well as any number of detectors.

\section{Discrimination probability} 
Consider a linear device that converts $D$ inputs into $D$ outputs. A single-mode input then becomes
\begin{align}
\hat a_\mathrm{in}^\dagger\rightarrow\sum_{j=1}^D\beta_j \hat b_j^\dagger,
\end{align}
where the vector with entries $\beta_j$ corresponds to a row of the corresponding transformation matrix.
If the device is balanced, we have $|\beta_j|^2=\frac1D$. A possible physical model for this device can be a cascaded sequence of $D-1$ conventional beamsplitters, with reflectivities $\frac{1}{D}$, $\frac{1}{D-1}$ \dots $\frac{1}{2}$, but other possibilities exist, including on-chip solutions. We note that all-optical solutions are just one area of applicability of our results, which can be applied to any multiplexer with a final set of detectors, which can be even as large as the set of pixels in an EMCCD or an ICCD. Configurations of the multiplexing part with closed paths are to be avoided, because the bosonic nature of photons would make them bunch and bypass the loops. For the same reason, we are not required to take phases into account. The multiplexer finally couples to a set of on/off single photon detectors. We wish to calculate the probability of observing $C$ clicks, given an initial photon number state of $N$ photons and given that all $D$ detectors have a quantum efficiency $\eta$ and a dark count probability $\varepsilon$. We start from the ideal case $\eta=1,\ \varepsilon=0$ and then move on to the general case $0\leq\eta\leq1,\ 0\leq\varepsilon\leq1$ and from the general case we retrieve a simple Corollary that holds for $0\leq\eta\leq1,\ \varepsilon\ll N/D$.

\subsection{Ideal detectors}
The fundamental ingredient for our analysis is the probability of distributing $N$ photons into exactly $C$ out of $D$ detectors. We start by numbering the detectors from 1 to $D$, then a certain string of numbers will describe an event, where the detectors numbered in the string are the ones that clicked. Note that in absence of noise the number of events cannot exceed the number of input photons, i.e. $C\leq N$. 

\begin{lemma*}
\label{lemma}
The probability of observing $C$ clicks by distributing a Fock state of $N$ photons evenly amongst $D$ ideal (i.e. noiseless and with 100\% quantum efficiency) on/off detectors is given by $$P_D(C|N)={D\choose C}\frac{C!}{D^N}\mathcal{S}_N^C,$$ where $\mathcal{S}$ is the Stirling number of second kind.
\end{lemma*}

\begin{proof}
Our goal is to compute the number of detection strings (i.e. the strings of numbers describing a detection event) that include exactly $C$ out of $D$ detectors.

Call $S_i$ the set of strings corresponding to $N$ input photons that do not include the $i$-th detector. Then select a specific subset $\mathcal K$ of cardinality $|\mathcal K|=k$ from the $D$ detectors. The set of strings that do not include any of the detectors in $\mathcal K$ is the intersection of the sets excluding each of the elements of $\mathcal K$: $\bigcap_{i\in \mathcal K}S_i$ and its cardinality is 
\begin{align}
\left|\bigcap_{i\in \mathcal K}S_i\right|=(D-k)^N
\end{align}
as we have $N$ choices with repetition, from $(D-k)$ possible detectors. Of course, we are also counting strings that exclude \emph{any other} detector, in addition to the ones in  $\mathcal K$. To get around this problem, we use the inclusion-exclusion rule to count the elements in unions of sets $S_i$. In particular, we need the union of $S_i$ for $i\in\{1,\dots,D\}$, i.e. the set of all strings that exclude \emph{at least} 1 detector, whose cardinality is
\begin{align}
\left|\bigcup_{i=1}^D S_i\right|=\sum_{j=1}^D(-1)^{j+1}{D\choose j}(D-j)^N
\end{align}
The complement of this set is the set of strings that include all $D$ detectors (if they missed any they would fall in $\bigcup_{i=1}^D S_i$), so by De Morgan's law we have
\begin{align}
\left|\overline{\bigcup_{i=1}^D S_i}\right|=\left|\bigcap_{i=1}^D\overline{S_i}\right|=\sum_{j=0}^D(-1)^{j}{D\choose j}(D-j)^N
\end{align}
Finally, we can compute the number of strings that include precisely $C$ out of $D$ detectors: pick $D-C$ detectors to be excluded (there are $D\choose C$ ways of doing this) and compute the number of strings that include all of the remaining $C$ detectors:
\begin{align}
{D\choose C}\left|\bigcap_{i=1}^C\overline{S_i}\right|&={D\choose C}\sum_{j=0}^{C}(-1)^{j}{C\choose j}(C-j)^N\\
&={D\choose C}C!\,\mathcal{S}_N^C,
\end{align}
where $\mathcal{S}_N^C$ is the Stirling number of the second kind.
So the probability of ending up with exactly $C$ clicks is the result above divided by the total number of possible strings $D^N$:
\begin{align}
P_D(C|N)&={D\choose C}\frac{C!}{D^N}\mathcal{S}_N^C
\label{stirling}
\end{align}
and our proof is complete.
\end{proof}

\subsection{Nonideal detectors}
Nonideal detectors are subject to mainly two effects: sub-unity quantum efficiency and noise, which can come from various sources. We model these as Bernoulli trials, where for each detector we have a probability $\eta$ of missing the photon and a probability $\varepsilon$ of a false count within the measurement window, in which case we learn that the detector clicked regardless of a photon hitting it or not (we don't worry about the source of noise, be it a dark count where the detector really fires albeit for no reason, or just electronic noise where we are informed of a click without it necessarily happening). Whether a detector detects an actual photon or gives a false count, we consider it out of order until the electronics have enough time to reset (e.g. about 40 ns for avalanche photodiodes). In this section we take both of these effects into account.
\begin{widetext}
\begin{theorem*}
The probability of observing $C$ clicks by distributing a Fock state of $N$ photons evenly amongst $D$ on/off detectors with quantum efficiency $\eta$ and false count probability $\varepsilon$ is given by $$P_{D,\eta,\varepsilon}(C|N)=\sum_{i=0}^C  p_\varepsilon(i|D)\sum_{j=C-i}^N  p_{\frac{D-i}{D}}(j|N)\sum_{k=C-i}^j p_\eta(k|j)P_{D-i}(C-i|k),$$
where $p_\xi(m|n)={n\choose m}\xi^m(1-\xi)^{n-m}$ is the probability of having $m$ successes out of $n$ trials when the success probability of a single trial is $\xi$.
\end{theorem*}
\end{widetext}

\begin{proof}
The proof comprises of 3 steps, each of which is of a similar nature: we consider in which ways an event can happen and we sum the relative probabilities. In the first step we split the observed number of clicks into spurious and real clicks. In the second step we split the initial photons into those that landed onto inactive detectors (the noisy ones) and those that landed onto active ones. In the third step we split the photons that landed onto active detectors into those that made it past the quantum efficiency and those that didn't. Finally, we use the ideal detection Lemma.

\emph{Step 1} We sum over the probability of obtaining $C$ total clicks by having $i$ of them come from noise and $C-i$ come from actual detections. We write the probability of $i$ false events given $D$ detectors as $p_\varepsilon(i|D)={D\choose i}\varepsilon^i(1-\varepsilon)^{D-i}$.

\emph{Step 2} Now $C-i$ clicks must come from real detection events from the remaining $D-i$ active detectors.
The probability that $j$ out of $N$ photons make to the $D-i$ active detectors is $p_{\frac{D-i}{D}}(j|N)$.

\emph{Step 3} As our detectors have a quantum efficiency $\eta\leq1$, the probability of remaining with $k$ out of $j$ photons is given by $p_\eta(k|j)$.

Now we can now apply the Lemma to write the probability of detecting $C-i$ out of $k$ survivor photons with $D-i$ detectors and combine these steps in the final result.
\end{proof}

There is a simple corollary of this theorem, which describes the case $\varepsilon=0$. Such corollary can be used even for noisy detectors as long as $D\varepsilon\ll N$:
\begin{corollary*}
The probability of observing $C$ clicks by distributing a Fock state of $N$ photons evenly amongst $D$ noiseless on/off detectors with quantum efficiency $\eta$ is given by $$P_{D,\eta}(C|N)=\sum_{k=C}^N p_\eta(k|N)P_{D}(C|k).$$
\end{corollary*}
\begin{proof}
We use the identity $p_0(m|n)=\delta_{m,0}$ to replace every occurrence of $i$ in the noisy detection Theorem by 0, and the identity $p_1(m|n)=\delta_{m,n}$ to replace every occurrence of $j$ by $N$. This gets rid of the first two summations and the result follows.
\end{proof}

\section{retrodicting the photon number}
To retrodict the photon number \emph{given} an observed number of clicks, we have to invert the probability in the main theorem using Bayes' rule:
\begin{align}
P_{D,\eta,\varepsilon}(N|C)=\frac{P_{D,\eta,\varepsilon}(C|N)Pr(N)}{\sum_k P_{D,\eta,\varepsilon}(C|k)Pr(k)}
\label{bayes}
\end{align}
This general formula is always valid, but it cannot be solved explicitly unless we specify the prior, which is what we will do next, for some special cases of particular relevance.

\subsection{Poisson prior}
In case of a Poissonian prior with mean photon number $\mu$ (which may occur when we deal with coherent states):
\begin{align}
Pr(N)=\frac{\mu^N e^{-\mu}}{N!},
\end{align}
we can find an explicit expression for the ideal retrodiction probability:
\begin{align}
P_{D}^\mathrm{Poisson}(N|C)=\frac{C!\,\mathcal{S}_N^C}{N!\,\gamma^{N}}\frac{1}{(e^{1/\gamma}-1)^{C}}
\end{align}
where $\gamma=D/\mu$.

\subsection{Thermal prior}
In case of a thermal prior with mean photon number $\mu$ (which occurs for instance for two-mode squeezed vacuum states or for EPR states)
\begin{align}
Pr(N)=\frac{\mu^N}{(\mu+1)^{N+1}},
\end{align}
the ideal retrodiction probability can be written as:
\begin{align}
\label{thermal}
P_{D}^\mathrm{Therm}(N|C)=\frac{C!\,\mathcal{S}_N^C}{(D+\gamma)^{N}}\frac{\Gamma(D+\gamma)}{\Gamma(D+\gamma-C)!}
\end{align}

\subsection{Considerations}
When one moves away from the ideal case, a sub-unity quantum efficiency plays a fundamental role, while the number of detectors is typically less important. One finds that the probability of detecting all the input photons with a noiseless apparatus, saturates at a value lower than 1 even for an infinite number of detectors:
\begin{align}
\lim_{D\rightarrow\infty}P_{D,\eta}(N|N)=\eta^N\lim_{D\rightarrow\infty}P_{D}(N|N)=\eta^N
\end{align}

The effect of noise in the detectors is tangible only when their number is sufficiently large, for instance when the number of spurious counts is comparable with the actual number of photons hitting the detectors i.e. when $D\varepsilon\approx N$.

\section{Applications}
We now would like to give a few examples of how to apply our results. The examples will be retrodiction of photon number for heralding quantum states.

\subsection{Example 1: heralding of a NOON state}
For this example we consider the following setup: we replace the two mirrors in the middle of a Mach-Zehnder (MZ) interferometer with 50:50 beam splitters and add detectors to measure the photons that leak. This configuration (if the phase difference between the two arms of the MZ is set to $\pi/2$) will output a $(|4,0\rangle+|0,4\rangle)/\sqrt{2}$ state if we start with the state $|3,3\rangle$ and if each of the two detectors measures exactly 1 photon.

Now the question is how well do we know that we had exactly 1 photon at the detectors? If we resort to multiplexed detection, we first need to compute the prior joint probability $Pr(N_1,N_2)$ of having $N_1$ photons at detector 1 and $N_2$ photons at detector 2. This is achieved using simple input-output relations for 50:50 beam splitters; we report it in Fig.\ref{table}.

\begin{figure}[h]
\begin{center}
\includegraphics[width=0.7\columnwidth]{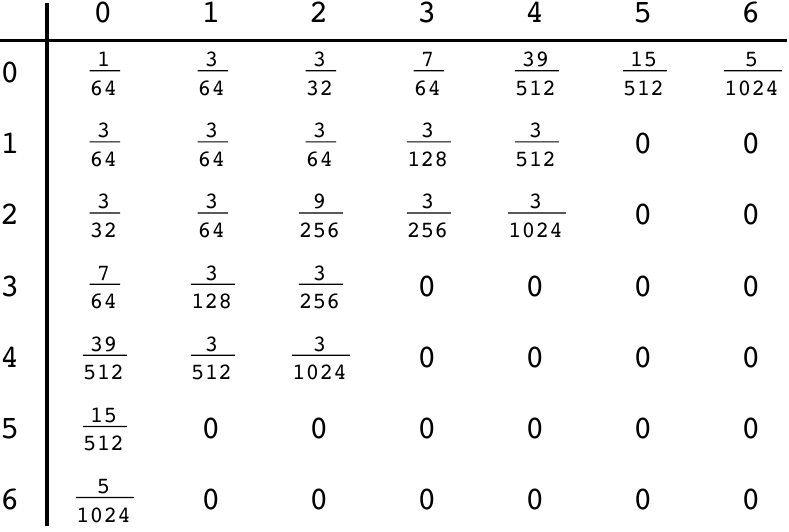}
\caption{\label{table}Joint probabilities of having $(i,j)$ photons (where $i$ and $j$ are listed in the headings on top and on the left) at the detectors in the modified MZ interferometer of the NOON state heralding example. These are computed assuming the input $|3,3\rangle$.}
\end{center}
\end{figure}

Then, we apply Bayes' rule (assuming that the two sets of multiplexed detectors are identical, but we could easily modify the equation below to account for different configurations) and find $P_{D,\eta,\varepsilon}(N_1,N_2|C_1,C_2)$ to be given by
\begin{align}
\frac{P_{D,\eta,\varepsilon}(C_1|N_1)P_{D,\eta,\varepsilon}(C_2|N_2)Pr(N_1,N_2)}{\sum_{k_1,k_2} P_{D,\eta,\varepsilon}(C_1|k_1)P_{D,\eta,\varepsilon}(C_2|k_2)Pr(k_1,k_2)}
\end{align}
We finally use the quantity $P_{D,\eta,\varepsilon}(N_1,N_2|C_1,C_2)$ to infer the retrodictive power of our multiplexed detectors. To complete the example, in Fig.~\ref{retroNOON} we plot the retrodicted probabilities of four configurations: 4 and 16 detectors with 60\% and 75\% quantum efficiency (and 500 dark counts/sec, with 10 ns gated measurement window), given that they both reported a single click each.

\begin{figure}[h]
\begin{tabular}{cc}
\subfloat[4 detectors, 60\% QE]{\includegraphics[width = .5\columnwidth]{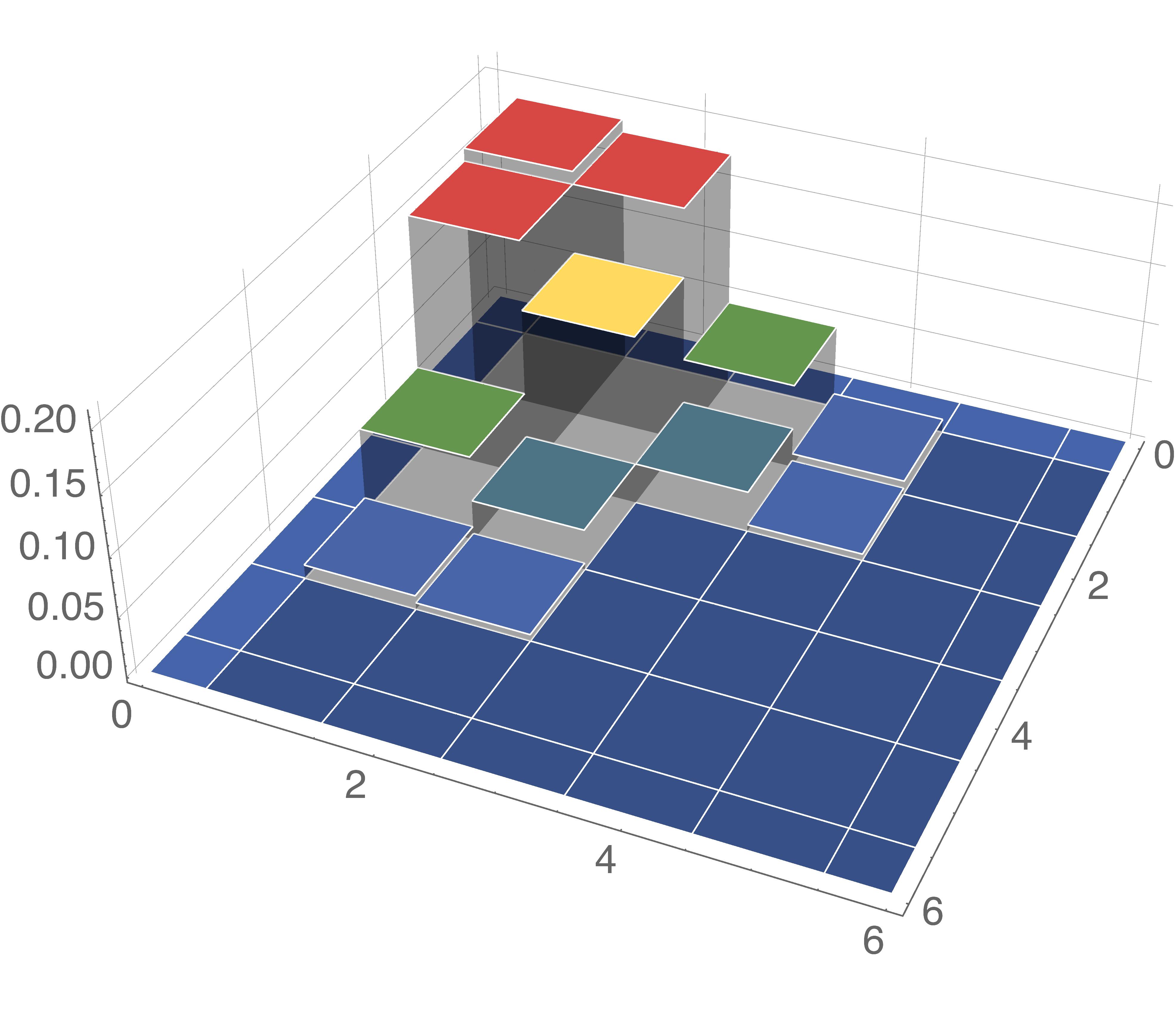}} &
\subfloat[4 detectors, 75\% QE]{\includegraphics[width = .5\columnwidth]{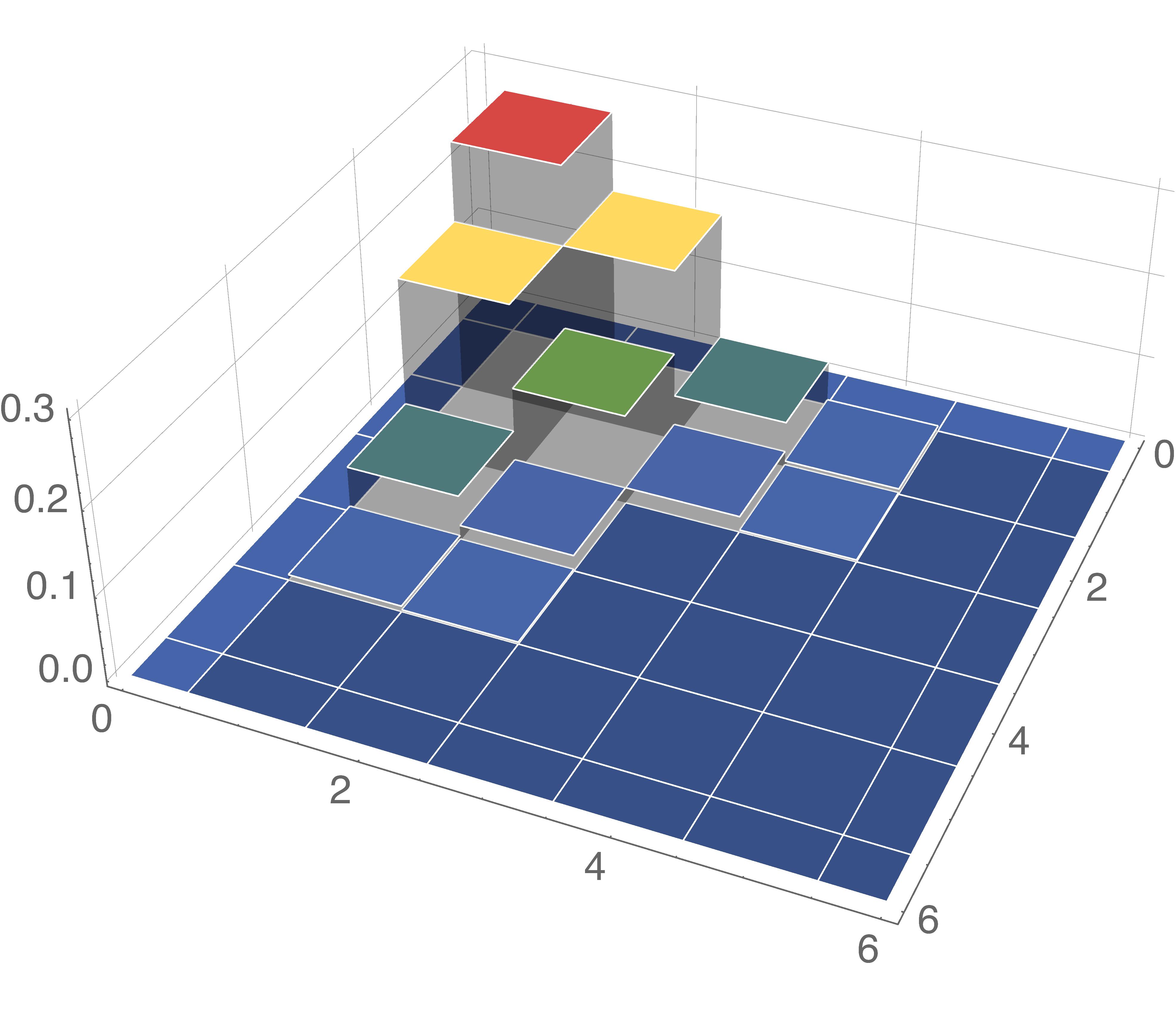}} \\
\subfloat[64 detectors, 60\% QE]{\includegraphics[width = .5\columnwidth]{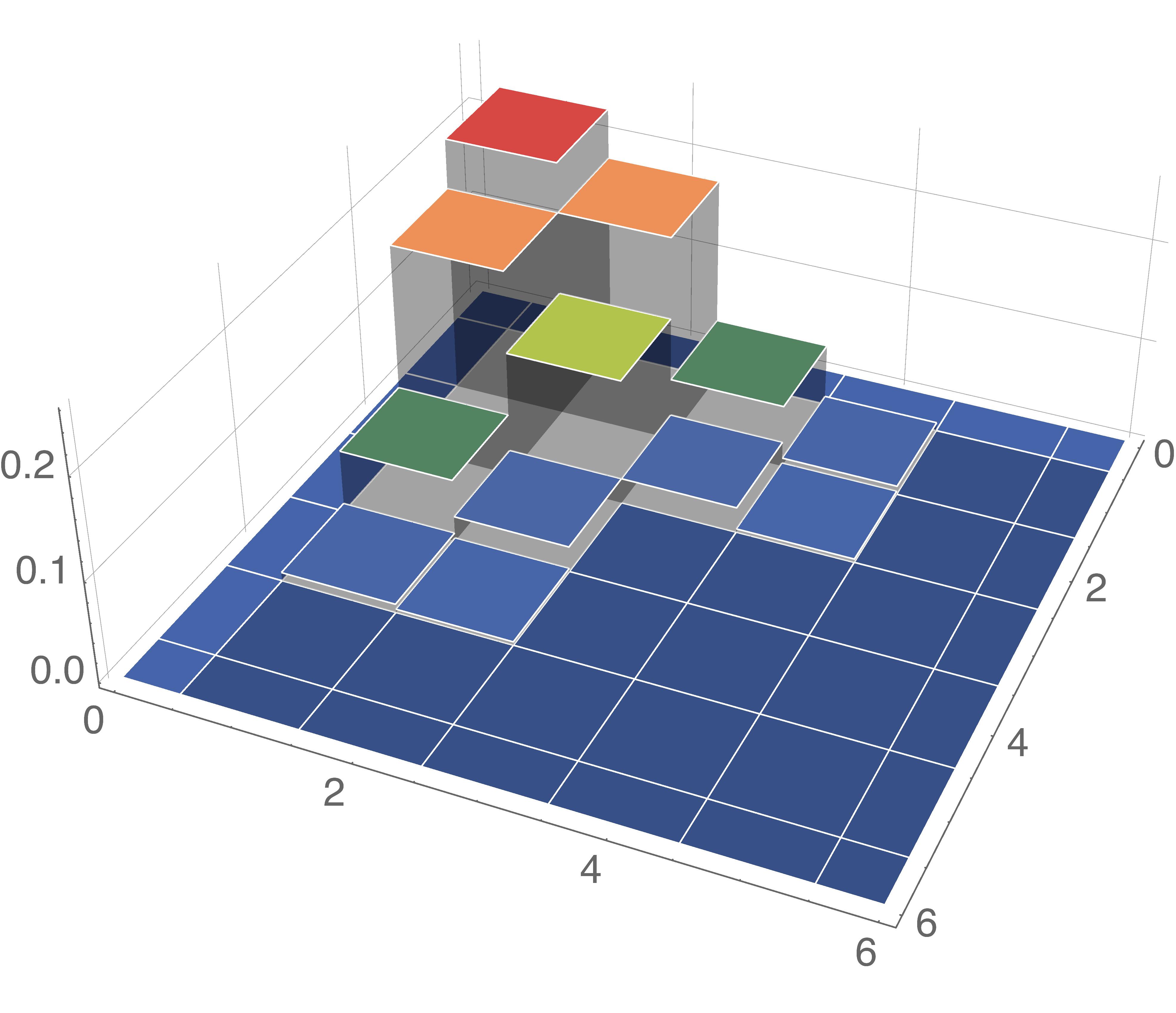}} &
\subfloat[64 detectors, 75\% QE]{\includegraphics[width = .5\columnwidth]{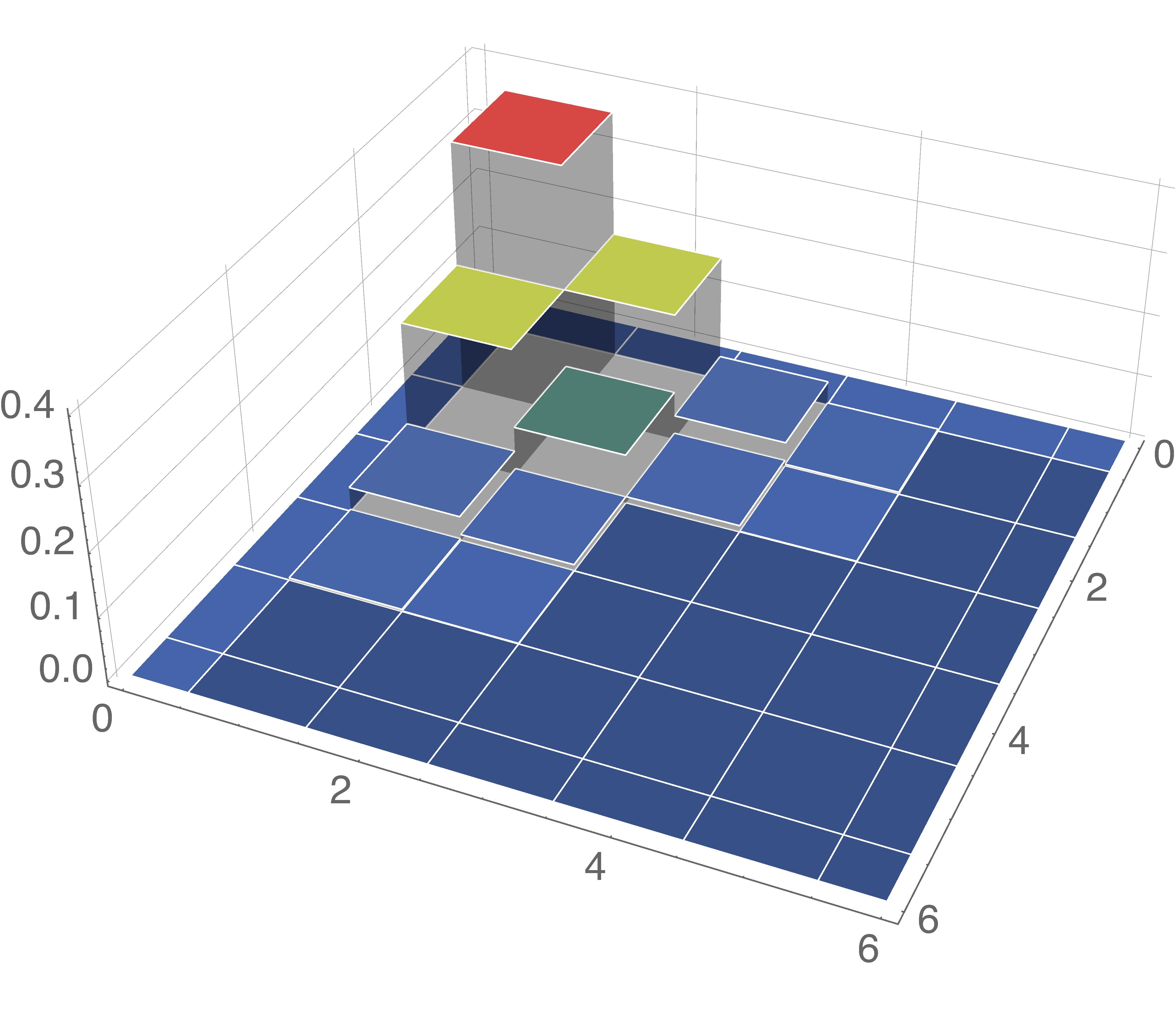}}
\end{tabular}
\caption{\label{retroNOON}Plots of the probability of retrodicted photon number for a NOON state heralding setup using multiplexed detection. Although the most probable case is the desired $|1,1\rangle$, its individual probability can be quite low, which leads to a low fidelity with the desired NOON state. The bottleneck in this case is quantum efficiency: even increasing the number of detectors from 4 to 64 does not perform as well as increasing the quantum efficiency from 60\% to 75\%.}
\end{figure}

For comparison, in Fig.~\ref{retroNOON2} we plot the retrodiction probabilities for a non-multiplexed measurement.
\begin{figure}[h!]
\begin{tabular}{cc}
\subfloat[1 detector, 75\% QE]{\includegraphics[width = .5\columnwidth]{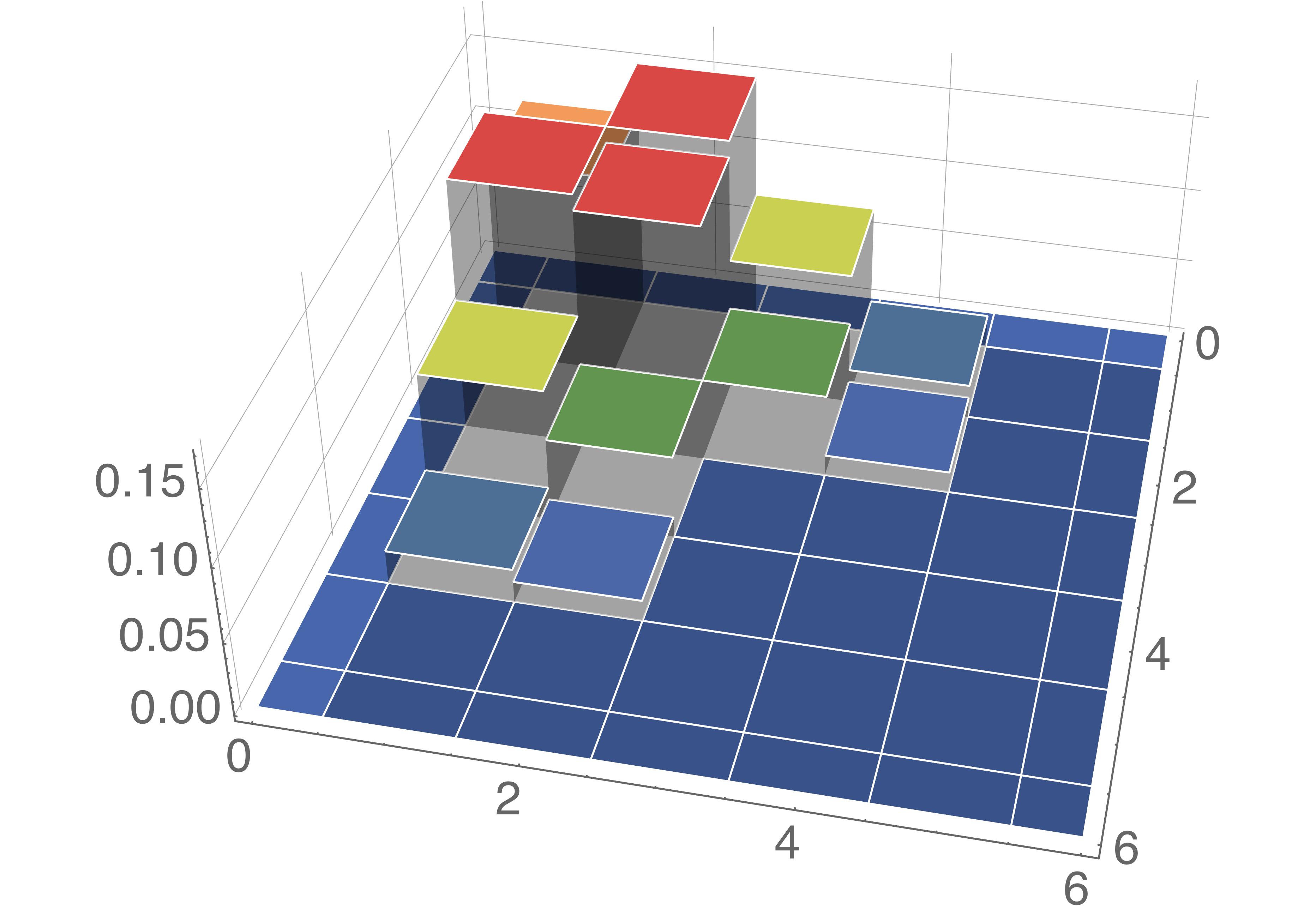}} &
\subfloat[1 detector, 100\% QE]{\includegraphics[width = .5\columnwidth]{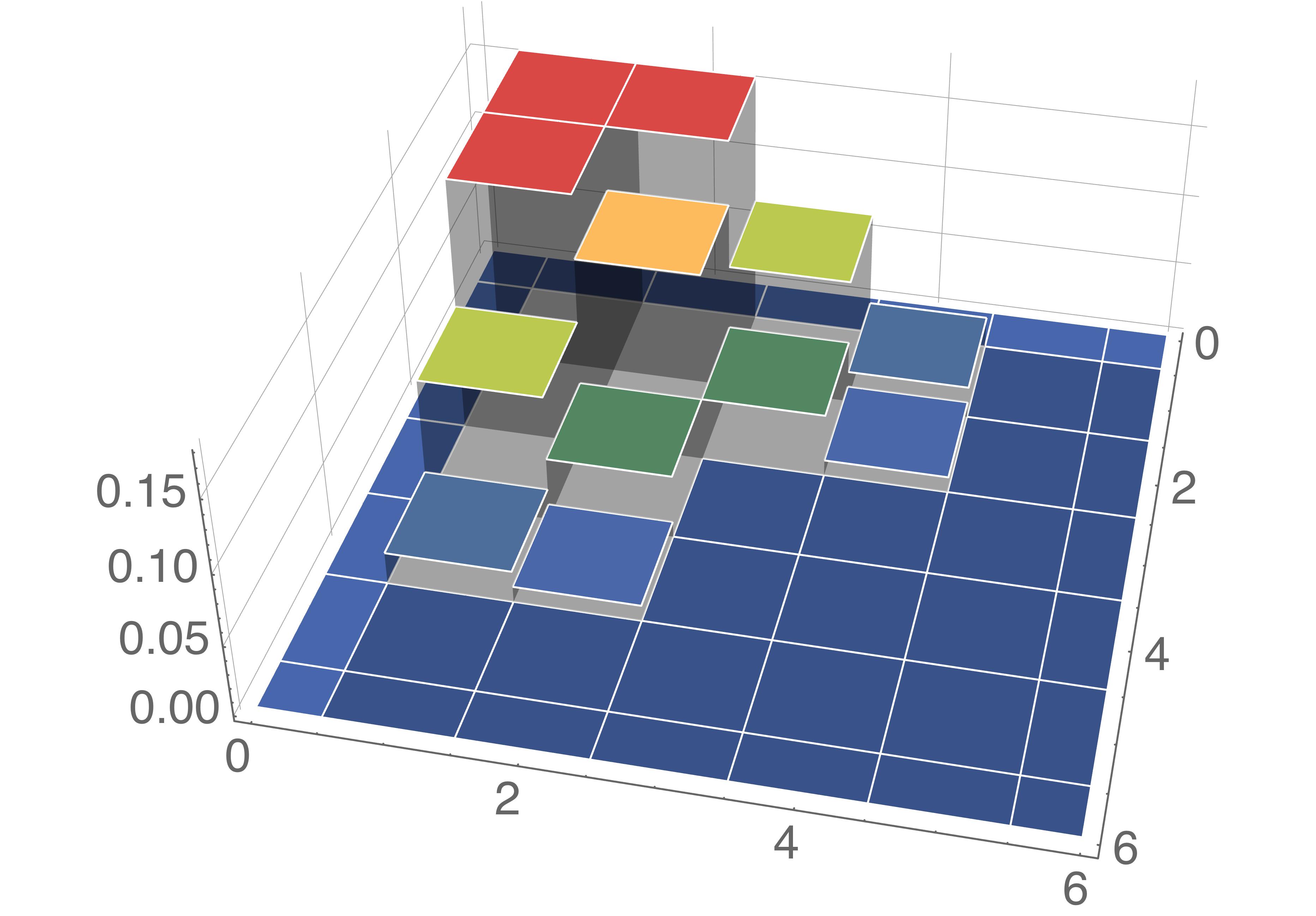}}
\end{tabular}
\caption{\label{retroNOON2}(left) A pair of realistic detectors are likely to lie: if they report a single click each, the state was more likely to be $|1,2\rangle$ or $|2,1\rangle$. (right) Even a pair of ideal (100\% quantum efficiency) detectors yields equal probability for the states $|1,1\rangle$, $|1,2\rangle$ and $|2,1\rangle$.}
\end{figure}

\subsection{Example 2: single photon heralding from squeezed vacuum}

We now consider an example of single photon heralding from a two-mode squeezed vacuum, which is performed by producing photons in pairs and heralding one by detecting the other. Such two-mode state can be generated by pumping a nonlinear crystal with an intense coherent laser pulse. The output of the process is a state in the following form:
\begin{align}
\hat S(\zeta)|0,0\rangle=\sum_{n=0}^\infty e^{in\phi}\frac{\sinh(g)^n}{\left(\sinh(g)^2+1\right)^{\frac{n+1}{2}}}|n,n\rangle,
\end{align}
where $\zeta=g e^{i\phi}$ is the squeezing parameter. For small enough values of the gain $g$ one can indeed ignore components with photon number larger than 1, but if the gain  is too large the heralded state can contain more than 1 photon. If such states were further used for crucial applications such as quantum cryptography, they would be vulnerable for example to the photon number splitting attack. Could a multiplexed detection scheme make for a better heralded single-photon source?
First note that the amplitudes of the two-mode squeezed vacuum follow a thermal distribution, if we recognize that $\sinh(g)^2$ is the mean photon number per mode. Then, we apply Eq.~\eqref{bayes} to find the retrodicted photon number distribution, which we plot for a few examples in Fig.~\ref{retroSqueezed}. Note that as the gain increases, the probability of the various number states levels off and becomes stable. 
\begin{figure}[h]
\begin{tabular}{cc}
\subfloat[4 detectors, 60\% QE]{\includegraphics[width = .5\columnwidth]{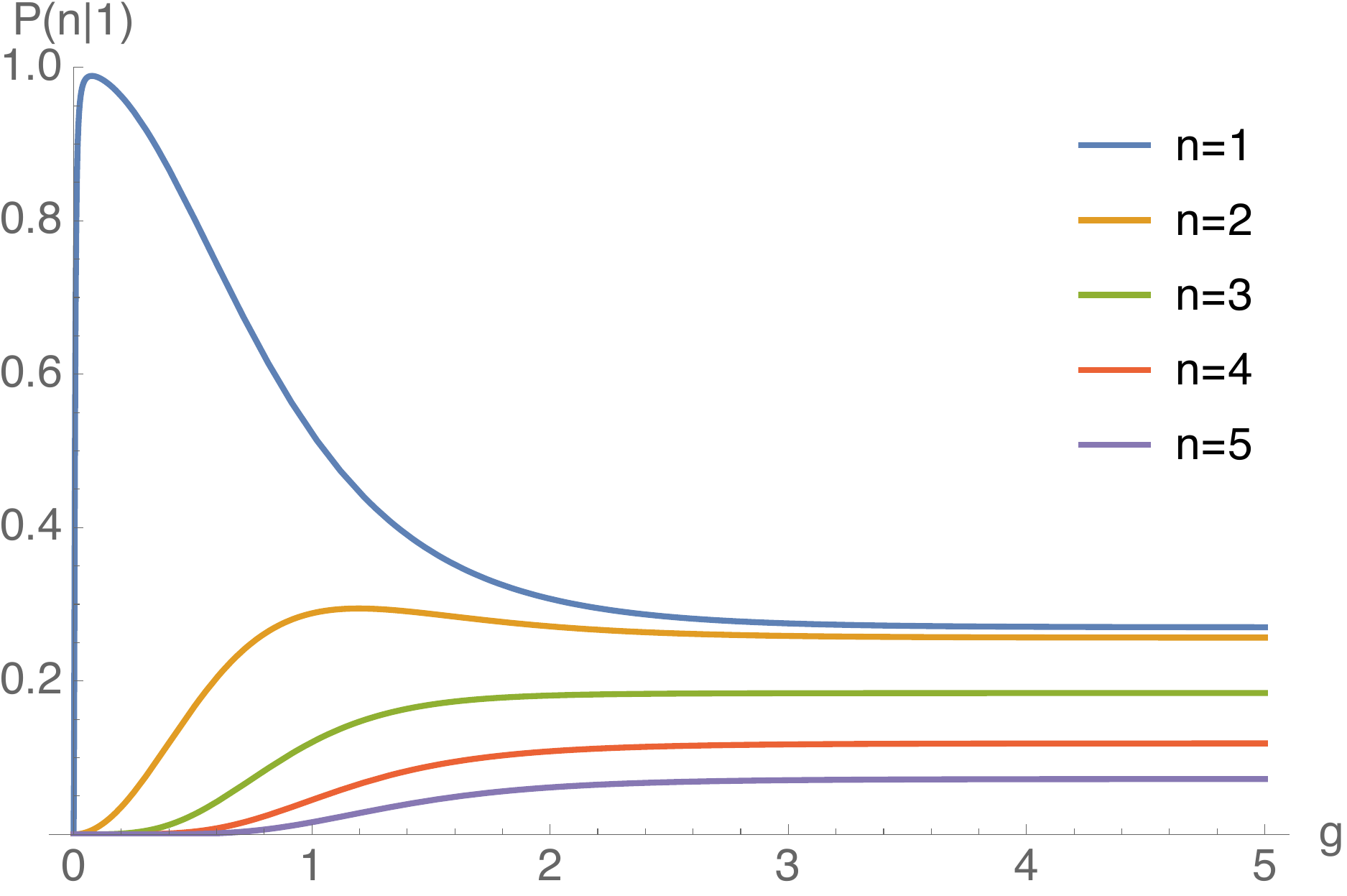}} &
\subfloat[4 detectors, 75\% QE]{\includegraphics[width = .5\columnwidth]{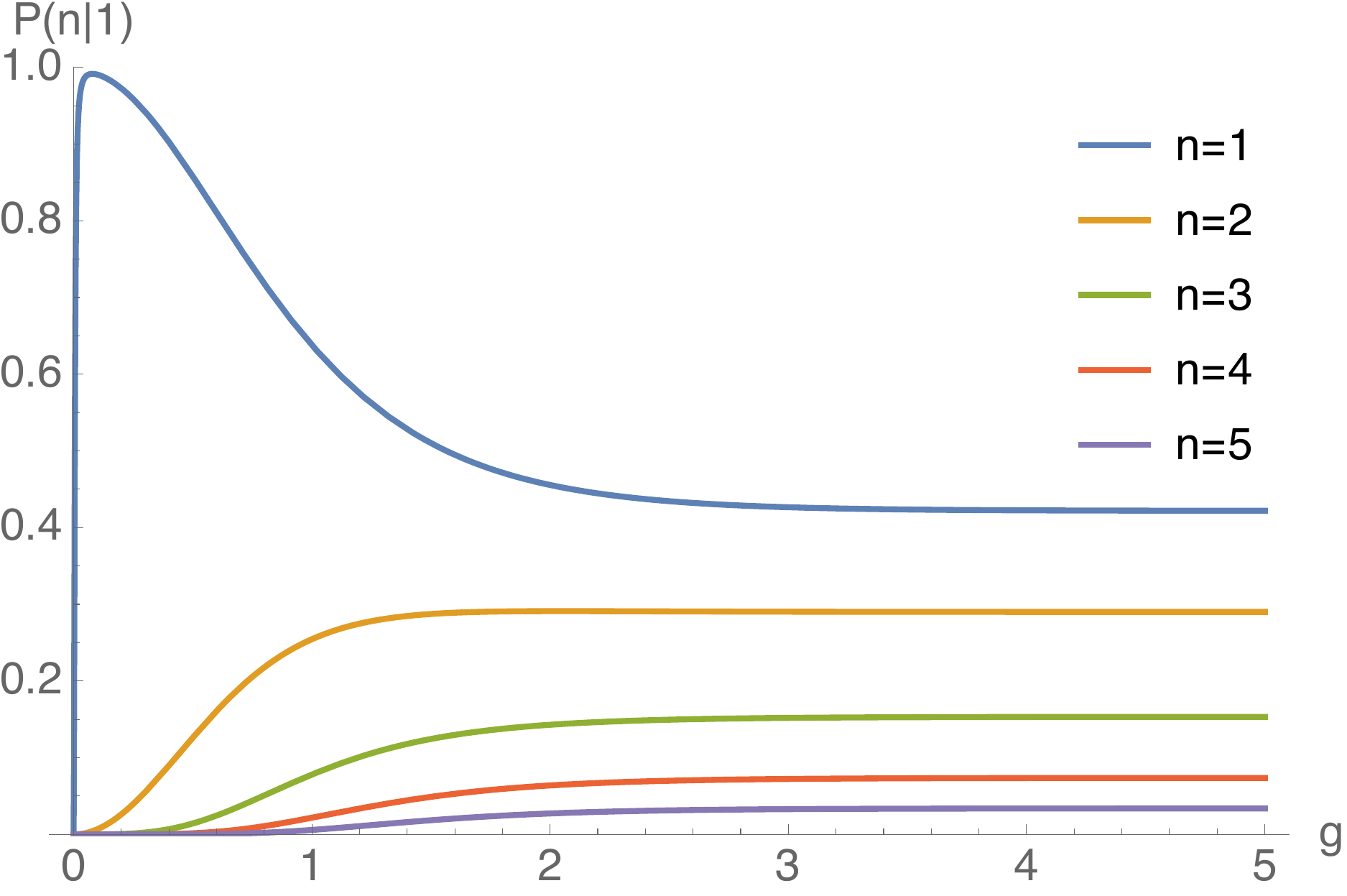}} \\
\subfloat[100 detectors, 60\% QE]{\includegraphics[width = .5\columnwidth]{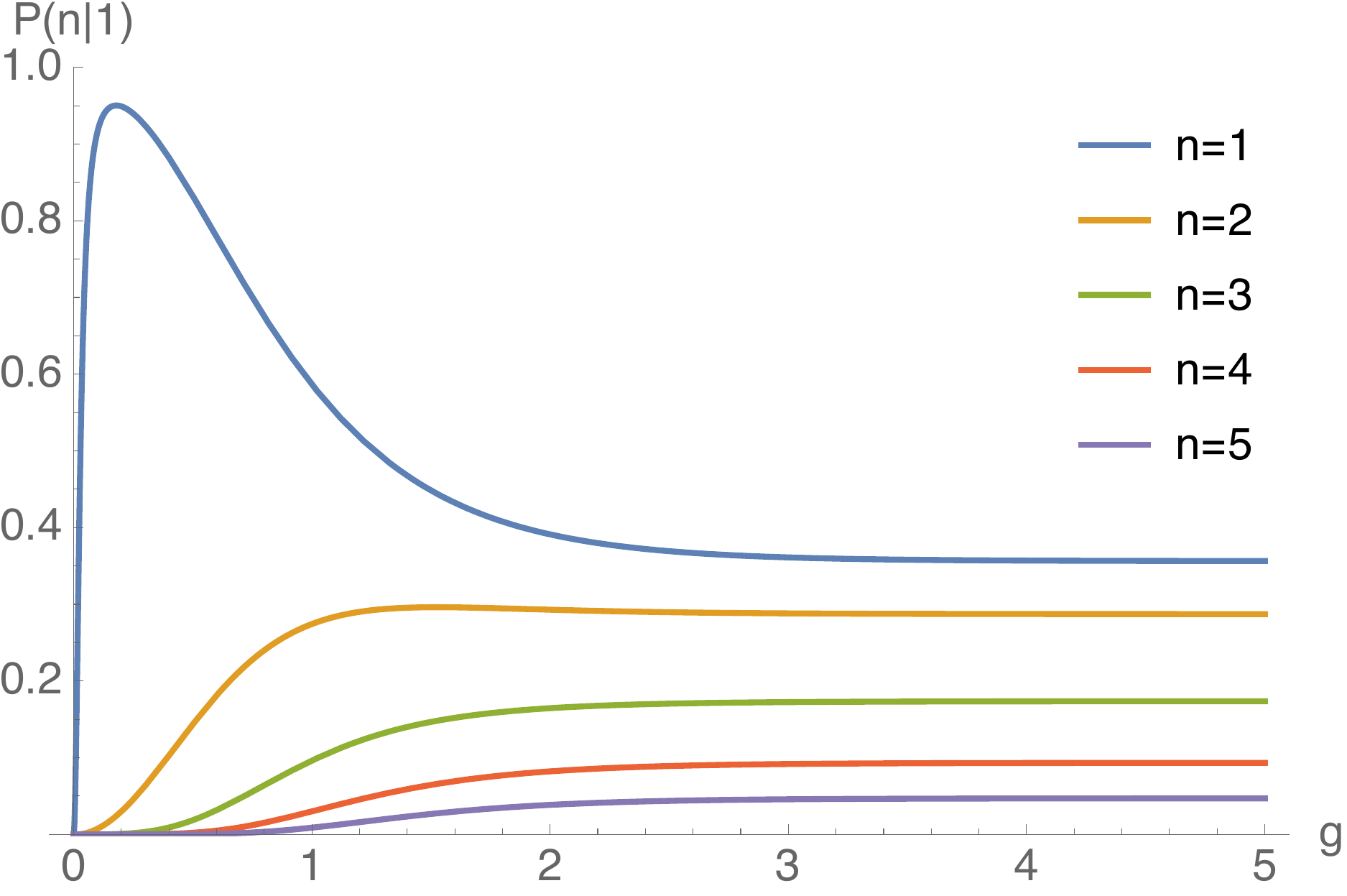}} &
\subfloat[100 detectors, 75\% QE]{\includegraphics[width = .5\columnwidth]{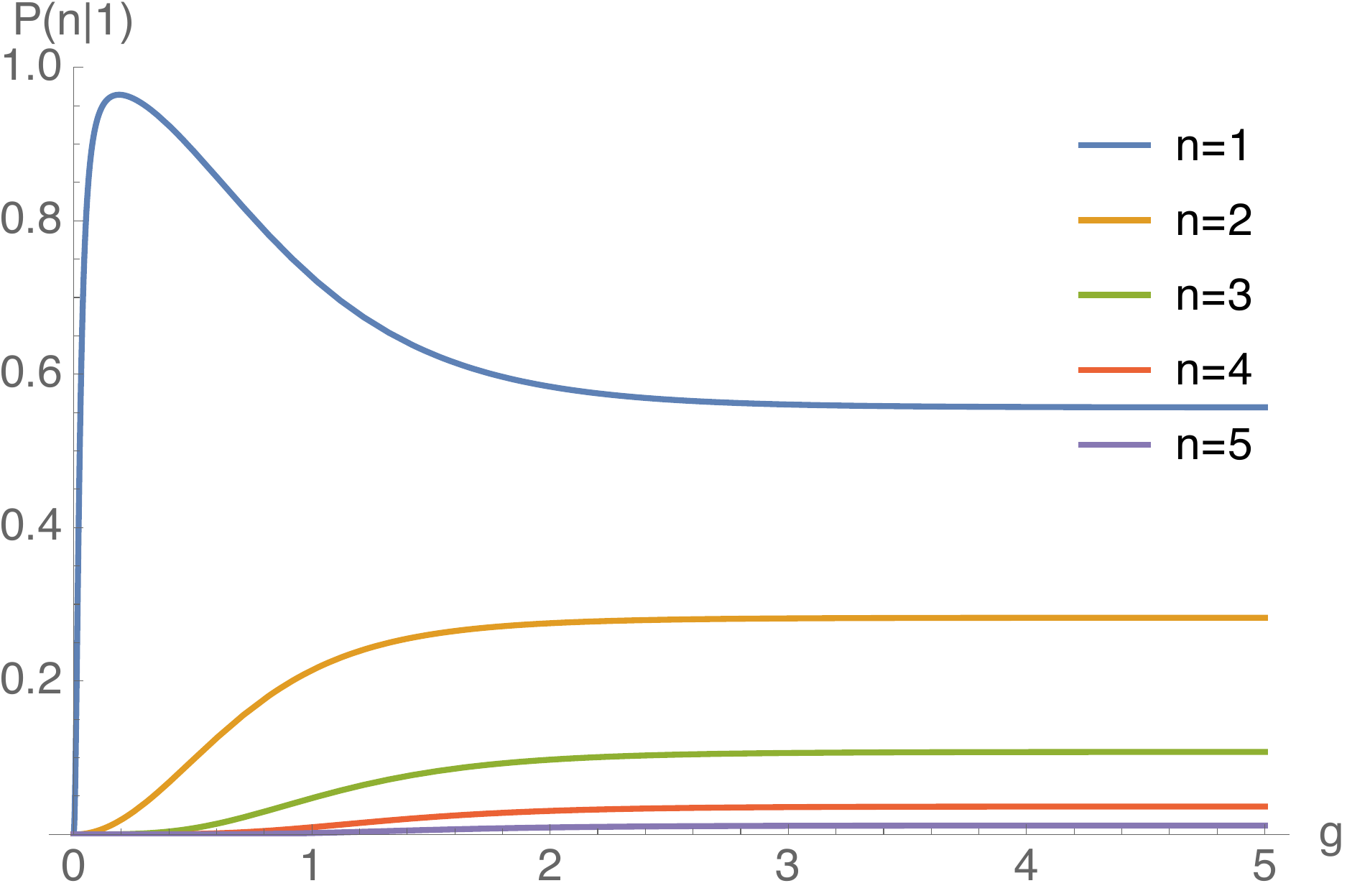}}
\end{tabular}
\caption{\label{retroSqueezed}Plots of the probability of retrodicted photon number for a squeezed vacuum state. Again, the bottleneck is quantum efficiency. Note that the probability of retrodicting a given photon number becomes constant as the gain $g$ increases.}
\end{figure}

\section{conclusions and outlook}
In conclusion, we have shown the most unbiased way of analyzing a detection event in a multiplexed measurement scheme, taking noise and efficiency into account. The corollary of our theorem can apply even to realistic situations if some conditions on the noise are met, which can be very advantageous as it is computationally much simpler to implement than the full theorem.

Our results can be applied also to optical engineering issues such as on-chip denoising in consumer imaging devices, where multiple pixels can fill an Airy disk and can be used to retrodict the intensity more accurately. There are still interesting questions to be asked, for instance whether it is possible to find closed form solutions of Eq.~\ref{bayes} for useful priors when the quantum efficiency is not unity, or if there is a reasonable way of relaxing the assumption of uniform illumination. We leave these to a future work.

\section{Acknowledgements}
This work was supported by the Canada Excellence Research Chairs program.

\bibliography{multiplexingBibliography}
\bibstyle{unsrt}

\end{document}